\documentclass[runningheads, envcountsame, a4paper]{article}

\usepackage{amssymb}
\usepackage{amsmath}
\usepackage{ amsfonts} 
\usepackage{ theorem}  
\usepackage{times}
\usepackage{stmaryrd}
\usepackage{mathrsfs}
\usepackage{latexsym}
\usepackage{hhline}

\usepackage{color}
\usepackage{graphicx}
\usepackage[skip=0pt]{caption}
\usepackage{subcaption}

\usepackage{tikz}
\usepackage{lineno}
\usepackage{todonotes}
\usepackage{bm}

\theoremstyle{plain}
\newtheorem{theorem}{Theorem}

\newtheorem{lemma}[theorem]{Lemma}

\newtheorem{proposition}[theorem]{Proposition}

\theoremstyle{definition}
\newtheorem{definition}[theorem]{Definition}
\newtheorem{example}[theorem]{Example}
\newtheorem{remark}[theorem]{Remark}

\newenvironment{proof}{\noindent{\bf Proof:\/}}{$\Box$\vskip 0.1in}

 \def\B{{\tilde{\delta}}}
 
\title{ {\Large \bf Isometric Words
{based on} 
\\
Swap and {Mismatch} Distance		}
\thanks{ Partially supported
by INdAM-GNCS Project 2022 and 2023, FARB Project  ORSA229894 
of University of Salerno, TEAMS Project of University of Catania and by  the MIUR Excellence Department Project MatMod@TOV awarded to the Department of Mathematics, University of Rome Tor Vergata.
 }
}
\author{Author One$^1$\thanks{Author One was partially supported by Grant XXX} \and Author Two$^2$ \and Author Three$^1$}

\date{
	$^1$Organization 1 \\ \texttt{\{auth1, auth3\}@org1.edu}\\%
	$^2$Organization 2 \\ \texttt{auth3@inst2.edu}\\[2ex]%
}

\author{ M. Anselmo$^1$ \and  G. Castiglione$^2$ \and  M. Flores$^1$ \and   D. Giammarresi$^3$ \and   M. Madonia$^4$ \and S. Mantaci$^2$
  }%

  \date{	
$^1${ \normalsize Dipartimento di Informatica,
Universit\`a di Salerno, 
Italy. 
\\
{\tt
\{manselmo, mflores \}@unisa.it}} \\
$^2${ \normalsize Dipartimento di Matematica e Informatica, Universit\`a di Palermo, 
Italy  
{\tt \{giuseppa.castiglione, sabrina.mantaci \}@unipa.it }} \\
$^3${ \normalsize Dipartimento di Matematica.
 Universit\`a  Roma ``Tor Vergata''
  Italy. 
{\tt giammarr@mat.uniroma2.it}}\\
$^4${\normalsize Dipartimento di Matematica eInformatica, Universit\`a  di
  Catania, 
  Italy.
  {\tt madonia@dmi.unict.it}} 
}

\begin{document}

\thispagestyle{empty}

 \maketitle

\begin{abstract}
An edit distance is a metric between words that quantifies how two words differ by counting the number of edit operations needed to transform one word into the other one. A word $f$ is said isometric with respect to an edit distance if, for any pair of $f$-free words $u$ and $v$, there exists a transformation of minimal length from $u$ to $v$ via the related edit operations such that all the intermediate words are also $f$-free. The adjective ``isometric'' comes from the fact that, if the Hamming distance is considered (i.e., only mismatches), then isometric words are connected with definitions of isometric subgraphs of hypercubes. 

We consider the case of edit distance with swap and mismatch. We compare it with the case of mismatch only  and prove some properties of isometric words that are related to particular features of their overlaps.

\end{abstract}


\section{Introduction}\label{s-intro}

The edit distance is a central notion in many fields of computer science.
It plays a crucial role in defining combinatorial properties of families of strings as well as in designing many classical string algorithms that find applications in natural language processing, bioinformatics and, in general, in information retrieval problems. The edit distance is a string metric that quantifies how two strings differ from each other and it is based on counting the minimum number of edit operations required to transform one string into the other one.

Different definitions of edit distance use different sets of edit operations. The operations of insertion, deletion and replacement of a character in the string characterize the Levenshtein distance which is probably the most widely known (cf. \cite{lev66}). On the other hand, the most basic edit distance is the Hamming distance which applies only to pair of strings of the same length and counts the positions where they have a mismatch; this corresponds to the restriction of using only the replacement operation. 
For this, the Hamming distance finds a direct application in detecting and correcting errors in strings and it is a major actor in the algorithms for string matching with mismatches (see \cite{galil90}).

The notion of isometric word (or string) combines the 
edit distance  with the property that a word does not appear as factor 
in other words.  Note that this property 
is important in combinatorics 
as well as in the investigation on similarities, or distances, on DNA sequences, where the avoided factor is referred to as an absent word \cite{BMR96,CMR20,CCFMP}. 
 Isometric words based on Hamming distance were first introduced 
 in \cite{IlicK12}
as special binary strings that never appear as factors in some string transformations. 
A string is $f$-\emph{free} if it does not contain $f$ as factor. A word $f$ is isometric if for any pair of $f$-free words $u$ and $v$, there exists a sequence of symbol replacement operations that transform $u$ in $v$ where all the intermediate words are also $f$-free.

Isometric words are connected with the definition of isometric subgraphs of the hypercubes, called generalized Fibonacci cubes.
The hypercube graph $Q_n$ is a graph whose vertices are the (binary) words of length $n$, and two vertices are adjacent when the corresponding words differ in exactly one symbol. Therefore, the distance between two vertices is the Hamming distance of the corresponding vertex-words. Let $Q_n(f)$ be the subgraph of $Q_n$ which contains only vertices that are $f$-free. Then, if $f$ is isometric, the distances of the vertices in $Q_n(f)$ are the same as calculated in the whole $Q_n$. Fibonacci cubes have been introduced by Hsu in \cite{Hsu93} and correspond to the case with $f=11$. 
In \cite{IlicK12,KlavzarS12,Wei17,WeiYZ19,Wei16} the structure of non-isometric words for alphabets of size 2 and Hamming distance is completely characterized and related to particular properties on their overlaps. The more general case of alphabets of size greater than 2 and Lee distance is studied in \cite{AFMWords21,MMM22,AnselmoFM22}.  
Using these characterizations, in \cite{BealC22} some linear-time algorithms are given to check whether a binary word is Hamming isometric and, for quaternary words, if it is Lee isometric. 
These algorithms were extended to provide further information on  non-isometric words, still keeping linear complexity in \cite{MMM22}.
Binary Hamming isometric two-dimensional words have been also studied in \cite{AGMS-DCFS20}.

Many challenging problems in correcting errors in strings 
come from computational biology. Among the chromosomal operations on DNA sequences, in gene mutations and duplication, it seems natural to consider  the {\em swap} operation, consisting in exchanging two adjacent symbols. 
The  Damerau-Levenshtein distance  adds also the swap to all edit operations.  In \cite{Wagner75}, Wagner proves that the edit distance with
insertion and swap is NP-hard, while each separate case  can be solved in polynomial time. Moreover, the general edit distance with 
insertion, deletion, replacement, and swap, is polynomially solvable. The 
swap-matching problem has been considered in \cite{AmirCHLP03,FaroP18}, and  algorithms for 
computing the corresponding edit distance are given in \cite{AmirEP06,DombbLPPT10}.

In this paper, we study the notion of isometric word using the edit distance based on swaps and mismatches. This distance will be referred to by using the {\em tilde} symbol that somehow evokes the swap operation. 
The tilde-distance ${\rm dist}_\sim(u,v)$ of
  equal-length words $u$ and $v$ is the minimum number of 
  replacement and swap operations to transform $u$ to into $v$. 
Then, the definition of {\em  tilde-isometric} word comes in a very natural way. A word $f$ is tilde-isometric if for any pair of equal-length words $u$ and $v$ that are $f$-free, there is a transformation from $u$ to $v$ that uses exactly dist$_\sim(u,v)$ replacement and swap operations and such that all the intermediate words still avoid $f$. It turns out that adding the swap operation to the definition makes the situation 
more complex, {but interesting for  applications}. It is not a mere generalization of Hamming string isometry since special situations arise. A swap operation in fact is equivalent to two replacements, but it counts as one when  computing the tilde-distance. Moreover, there could be different ways to transform $u$ into $v$ since particular triples of consecutive symbols can be managed, from left to right, either by first a swap and then a replacement or 
by a replacement and then a swap. 
We present some examples of tilde-isometric words that are not Hamming isometric and vice versa.  By definition, in order to prove that a given string $f$ is not  tilde-isometric one should exhibit a pair of $f$-free words $(\tilde\alpha, \tilde\beta)$ such that any transformation from $\tilde\alpha$ to $\tilde\beta$ of length dist$_\sim(\tilde\alpha, \tilde\beta)$ comes through words that contain $f$. Such a pair is called pair of {\em tilde-witnesses} for $f$.  We prove  some necessary conditions for $f$ to be non-isometric based on the notion of error-overlap and give an explicit construction of the tilde-witnesses in many cases.

\section{Preliminaries}
Let $\Sigma$ be a finite alphabet.
A word (or string) $w$ of length $|w|=n$, is $w=a_1a_2\cdots a_n$, 
where $a_1, a_2, \ldots, a_n$ are symbols in $\Sigma$. The set of all words over $\Sigma$ is denoted $\Sigma^*$ and the set of all words over $\Sigma$ of length $n$ is denoted $\Sigma^n$. Finally, $\epsilon$ denotes the {\em empty word} and $\Sigma^+=\Sigma^* - \{\epsilon\}.$
For any word $w=a_1a_2\cdots a_n$, the {\em reverse} of $w$ is the word $w^{rev}=a_na_{n-1}\cdots a_1$. If $x \in \{0, 1\}$, we denote by $\overline{x}$ the opposite of $x$, i.e $\overline{x}=1$ if $x=0$ and viceversa. 
Then we define {\em complement} of $w$ the word $\overline{w}=\overline{a}_1\overline{a}_2\cdots \overline{a}_n$.

Let 
  $w[i]$ denote the symbol of $w$ in position $i$, i.e. $w[i]=a_i$.
  Then, $w[i .. j] = a_i \cdots a_j$,
   for $1\leq i\leq j\leq n$, is a \emph{factor}  of $w$. 
  The \emph{prefix}  (resp. \emph{suffix}) of $w$ of length $l$, with $1 \leq l \leq n-1$ is ${\rm pre}_l(w) = w[1 .. l]$  (resp. ${\rm suf}_l (w) = w[n-l+1 .. n]$).
  When ${\rm pre}_l(w) = {\rm suf}_l (w)=u$ then $u$ is here referred to as an \emph{overlap} of $w$ of length $l$; 
  it is also called border, or bifix. A word $w$ is said {\em $f$-free} if $w$ does not contain $f$ as a factor.

An {\em edit operation} is a function $O: \Sigma^* \to \Sigma^*$ that transform a word into another one. Among the most common edit operations there are the insertion, the  deletion or the replacement of a character and the swap of two adjacent characters. 
Let $OP$ be a {\em set of edit operations}. The {\em edit distance} of two words $u, v \in\Sigma^*$ 
is the minimum number of edit operations in $OP$ needed to transform $u$ into $v$. .  
In this paper, we consider the edit distance that uses only replacements and swaps. Note that these two operations do not change  the length of the word. We give a formal definition.
\begin{definition}
    Let $\Sigma$ be a finite alphabet and $w=a_1a_2\ldots a_n$ a word over $\Sigma$.
    \\The {\em replacement operation} (or {\em replacement}, for short) on $w$ at position $i$ with $x\in \Sigma$, $x\neq a_i$,
is defined by
   $$R_{i,x}(a_1a_2\ldots a_{i-1}a_ia_{i+1}\ldots a_n)=a_1a_2\ldots a_{i-1} x a_{i+1}\ldots a_n.$$
    \noindent The {\em swap operation} (or {\em swap}, for short) on $w$ at position $i$ consists in exchanging  characters at positions $i$ and $i+1$, provided that they are different, $a_i \neq a_{i+1}$, 
    $$S_i(a_1a_2\ldots a_ia_{i+1}\ldots a_n)=a_1a_2\ldots a_{i+1}a_i \ldots a_n.$$
     \noindent 
    \end{definition}
When the alphabet $\Sigma=\{0,1\}$ there is only a possible replacement at a given position $i$, so we  write $R_i(w)$ instead of $R_{i,x}(w)$. 

Given two equal-length words $u=a_1\cdots a_n$ and $v=b_1\cdots b_n$, they have a {\em mismatch error} (or {\em mismatch}) at position $i$ if $a_i\neq b_i$ and they have a {\em swap error} (or {\em swap}) at position $i$ if $a_ia_{i+1}=b_{i+1}b_i$, with $a_i \neq a_{i+1}$. We say that $u$ and $v$ have an {\em error} at position $i$ if they have either a mismatch or a swap error. 

Note  that one swap corresponds to two adjacent mismatches. 

A word $f$ is isometric if for any pair of $f$-free words $u$ and $v$, there exists a 
sequence of minimal length of replacement operations that transform $u$ into $v$ where all the intermediate words are also $f$-free. 
 In this paper we refer to this definition of isometric as {\em Ham-isometric}.
In \cite{WYZ19}, 
a word $w$ has a $2$-error overlap  if  there exists $l$   such that 
${\rm pre}_l (w)$ and ${\rm suf}_l (w)$ have two mismatch errors. Then, they prove  
the following characterization. 
\begin{proposition}\label{p-Hamiso}
   A word $f$ is Ham-isometric if and only if $f$ has a 2-error overlap.
\end{proposition}

\section{Tilde-distance and tilde-isometric words}\label{s-tildeiso}

In this section we consider the edit distance based on swap and replacement operations that we call tilde-distance and we denote $dist_\sim$. 
First, we give some  definitions and notations, together with  some examples and the proofs of  some preliminary properties. 

\begin{definition}\label{def:d-tilde-distance}
	Let $u, v \in \Sigma^*$ be words of equal length.
	The {\em tilde-distance} ${\rm dist}_\sim(u,v)$ between $u$ and $v$ is the minimum number of   replacements and swaps needed to transform $u$ into $v$.
	\end{definition}
\begin{definition}\label{def:tilde-transformation}
  	Let $u, v \in \Sigma^
   *$ be words of equal length.  
  	A \emph{tilde-transformation} $\tau$
  	of length $h$ from $u$ to $v$ is a sequence of words $(w_0, w_1, \ldots, w_h)$ such that
  	$w_0=u$, $w_h=v$, and for any $k=0, 1, \ldots ,h-1$,
  	${\rm dist}_\sim(w_k,w_{k+1})=1$.  Moreover, $\tau$ is {\em $f$-free} if for any $i = 0,1, \ldots ,h$, the word $w_i$ is $f$-free. 
  \end{definition}
    A tilde-transformation $(w_0, w_1, \ldots, w_h)$ from $u$ to $v$  {is associated to a} sequence of $h$ operations $(O_{i_1}, O_{i_2},\ldots O_{i_h})$  such that, for any $k=1, \ldots ,h$,  $O_{i_k} \in \{R_{i_k,x},S_{i_k}\}$ and $w_{k}=O_{i_k}(w_{k-1})$; 
 it can be represented as follows: 
 $$u=w_0\xrightarrow{O_{i_1}}w_1 \xrightarrow{O_{i_2}} \cdots \xrightarrow{O_{i_h}}w_h=v.$$
 With a little abuse of notation,  in the sequel we will refer to a tilde-transformation both  as a sequence of words and as a sequence of operations. 
We give some examples.
\begin{example}\label{ex:tilde-trans} 
Let  $u=1011, v=0110$. Below we show two different tilde-transformations from $u$ to $v$.  Note that the length of $\tau_1$ corresponds to ${\rm dist}_\sim(u,v)=2$.
$$\tau_1: 1011\xrightarrow{S_1}0111 \xrightarrow{R_4}0110\hspace{1cm}\tau_2:1011 \xrightarrow{R_1} 0011 \xrightarrow{R_2}0111 \xrightarrow{R_4}0110$$
Furthermore, consider the  following tilde-transformations of $u'=100$ into $v'=001$:
$$\tau'_1:100\xrightarrow{S_1}010 \xrightarrow{S_2}001 \hspace{1cm}\tau'_2:100 \xrightarrow{R_1} 000  \xrightarrow{R_2}001$$
 Note that both $\tau'_1$ and $\tau'_2$ have the same length equal to ${\rm dist}_\sim(u',v')=2$. Interestingly, in $\tau'_1$ the symbol in position 2 is changed twice.
\end{example}

The next lemma shows that, in the case of a two letters alphabet, we can restrict to  tilde-transformations where each character is changed at most once.

\begin{lemma}\label{l-no-sovrapp}
 Let $u,v\in \{0,1\}^m$  with $m\geq 1$. Then, there exists a tilde-transformation of $u$ into $v$ of length ${\rm dist}_\sim(u,v)$ such that for any $i=1, 2, \dots , m$, the character in position $i$ is changed at most once.
\end{lemma}

\begin{proof} Let  $u=a_1\cdots a_m$ and $v=b_1\cdots b_m$ and let $\tau$
be a tilde-transformation of $u$ into $v$ of length $d={\rm dist}_{\sim}(u,v)$. 
Suppose  that, for some $i$, the character in position $i$ is changed more than once by $\tau$ and let $O_t$ and $O_s$ be the first and the second operation, respectively, that modify the character in position $i$.
Observe that the character in position $i$ can be changed by the operations $R_i$, $S_{i-1}$ or $S_i$. 
\\
 Suppose  that $O_t=S_i$ and $O_s=R_i$. Then, the symbol $a_i$ is changed twice and two operations $S_i$ and $R_i$ could be replaced by a single $R_{i+1}$. This would yield a tilde-transformation of $u$ into $v$ of length strictly less than $d$; this is a contradiction to the definition of tilde-distance. 
 Similarly for the cases where $O_t=R_i$ and $O_s=S_i$, $O_t=S_{i-1}$ and $O_s=R_i$, $O_t=R_i$ and $O_s=S_{i-1}$. 
 \\
  Finally, if 
  $O_t=S_{i-1}$ and $O_s=S_i$
   then the three characters in positions $i-1$, $i$ and $i+1$ are changed, but the one in position $i$ is changed twice. Hence, the two swap operations $S_{i-1}$ and $S_i$ can be replaced by $R_{i-1}$ and $R_{i+1}$ yielding a tilde-transformation of $u$ into $v$ of same length $d$ which instead involves positions $i-1$ and $i$ just once (see $\tau'_2$ in  Example  \ref{ex:tilde-trans}). 
   \end{proof}

\begin{remark}
Lemma \ref{l-no-sovrapp} only applies to a binary alphabet. Indeed, if $\Sigma=\{0,1,2\}$, and take $u=012$ and $v=120$, then ${\rm dist}_\sim(012, 120)=2$ because there is the tilde-transformation  $012 \xrightarrow{S_1} 102 \xrightarrow{S_2} 120$. Instead, in order to change  each character at most once, three replacement operations are needed.
\end{remark}

\begin{definition}
 Let $\Sigma$ be a finite alphabet and $u,v\in \Sigma^+$. A tilde-transformation from $u$ to $v$ is {\em minimal} if its length is equal to ${\rm dist}_{\sim}(u,v)$ and characters in each position are modified at most once.   
\end{definition}

Lemma \ref{l-no-sovrapp} guarantees that, in the binary case, a minimal tilde-transformation always exists. In the sequel, this will be the most investigated case.
Let us now define isometric words based on the swap and mismatch distance.

\begin{definition}\label{d-tilde-iso}
	Let $f\in\Sigma^n$, with $n\geq 1$, $f$ is \emph{tilde-isometric} if for any pair of $f$-free words $u$ and $v$ of length $m>n$, there exists a minimal tilde-transformation from $u$ to $v$ that is $f$-free.
    It is \emph{tilde-non-isometric} if it is not tilde-isometric. 
\end{definition}

In order to prove that a word is tilde-non-isometric it is sufficient to exhibit a pair $(u,v)$ of words contradicting the Definition \ref{d-tilde-iso}.  Such a pair will be referred to as  tilde-witnesses for $f$. 
Some examples follow.

\begin{definition}\label{d-witnesses}
    A pair $(u,v)$ of words in $\Sigma^m$ is a pair of \emph {tilde-witnesses} for $f$ 
    if: 
    
    1. $u$ and $v$ are $f$-free
    
    2. ${\rm dist}_{\sim}(u,v) \geq 2$
    
      3. there exists no minimal tilde-transformation from $u$ to $v$ that is $f$-free.
\end{definition}

\begin{example}\label{e-1010}
The word $f=1010$ is tilde-non-isometric because $u=11000$ and $v=10110$ are tilde-witnesses for $f$. In fact, the only possible minimal tilde-transformations from $u$ to $v$ are
$11000 \xrightarrow{S_2} 10100 \xrightarrow{R_4} 10110$ and $11000 \xrightarrow{R_4} 11010 \xrightarrow{S_2} 10110$ and in both cases $1010$ appears as factor after the first step.
\end{example}

\begin{remark}\label{r-ham-replacements}
When a transformation contains a swap and a replacement that are adjacent, there could exist many distinct minimal tilde-transformations that involve different sets of operations. 
For instance, the pair $(u,v)$, with $u=010$ and $v=101$, has the following minimal tilde-transformations:
$$010 \xrightarrow{S_1} 100 \xrightarrow{R_{3}} 101 \hspace{1cm}010 \xrightarrow{S_2} 001 \xrightarrow{R_{1}} 101$$
 This fact cannot happen when only replacements are allowed. For this reason studying tilde-isometric words is more complicated than the Hamming case. 
\end{remark}

Example \ref{e-1010} shows a tilde-non-isometric word. Proving that a given word is tilde-isometric is much harder since it requires to give evidence that no tilde-witnesses exist. We will now prove that word $111000$ is isometric with {\em ad-hoc} technique.

\begin{example}\label{ex-111000-good}
The word $f=111000$ is tilde-isometric. Suppose by the contrary that  $f$ is tilde-non-isometric and let $(u,v)$ be a pair of tilde-witnesses for $f$ of minimal tilde-distance. 
If $u$ and $v$ have only mismatch errors, this is the case of the Hamming distance and results from this theory \cite{MMM22,Wei17} show that 
$u=1110\bm{01}1000$ and $v=1110\bm{10}1000$; these are not tilde-witnesses since ${\rm dist}_\sim(u,v)=1.$
\\
Therefore, $u$ and $v$ have a swap error in some position $i$; suppose $u[i..i+1]=01$.
The minimality of ${\rm dist}_\sim(u,v)$ implies that 
$S_i(u)$ is not $f$-free. Then, a factor $111000$
appears in $S_i(u)$ from position $i-2$, and $u[i-2 ..i+3]= 110100$. 
Since $v$ is $f$-free, then there is another error in $u$ involving some positions in $[i-2 ..i+3]$.
It cannot be neither a swap 
(since there are no adjacent different symbols that are not changed yet), nor a mismatch in positions $i-1$, $i+2$ (since the corresponding replacement cannot let $f$ occur). Then, it is a mismatch in position $i-2$ or $i+3$.
Consider the case of a mismatch in position $i+3$ (the other case is analogous). 
Then, $u[i+3 ..i+8]= 011000$ and there is another error in $[i+4 ..i+8]$, in fact, in position $i+6$ or $i+8$. Continuing with similar reasoning, one  falls back  to the previous situation.
This is  a contradiction because
the length of $u$ is finite.
\end{example}

From now on,  
we consider only the binary alphabet $\Sigma=\{0, 1\}$ and
we study isometric binary words beginning by $1$, in view of the following lemma whose  proof can be easily inferred by combining the definitions.


\begin{lemma}\label{l-rev-comp}
	Let $f\in\{0,1\}^n$. The following statements are equivalent:
 \begin{enumerate}
     \item $f$ is tilde-isometric
     \item $f^{rev}$ is tilde-isometric 
     \item $\overline{f}$  is tilde-isometric.
 \end{enumerate}
\end{lemma}
  
Let us conclude the section by comparing tilde-isometric with Ham-isometric words.
Although the tilde-distance is more general than the Hamming distance, they
are incomparable, as stated in the following proposition.

\begin{proposition}\label{p-compareHam}
There exists a word which is tilde-isometric but Ham-non-isometric, and
a word which is tilde-non-isometric, but Ham-isometric.
\end{proposition}

\begin{proof}
The word $f=111000$ is tilde-isometric (see Example  \ref{ex-111000-good}), but $f$ is Ham-non-isometric by Proposition \ref{p-Hamiso}. 
	\\	
Conversely, $f'=1010$ is tilde-non-isometric (see Example \ref{e-1010}), but Ham-isometric by Proposition \ref{p-Hamiso}.

\end{proof}

\section{Tilde-isometric words and tilde-error overlaps}\label{s-iso-e-2eo}

 {In this section we focus on the word property of being tilde-non-isometric and connect it to the number of errors in its overlaps. The idea reminds the characterization for Ham-isometric words recalled in Proposition \ref{p-Hamiso} but the swap operation changes all the perspectives as pointed also in Proposition \ref{p-compareHam}. }

\begin{definition}
	Let $f\in \{0,1\}^n$.
	Then, $f$ has a {\em $q$-tilde-error overlap} of length $l$,  with
	$1 \leq l \leq n-1$ and $0\leq q \leq l$,
	if
	${\rm dist}_{\sim} ({\rm pre}_l(f), {\rm suf}_l(f))=q$.
\end{definition}

\begin{figure}[b]
\begin{center}
\hspace{-1.6cm}
    \begin{subfigure}{0.4\textwidth}
\includegraphics[width=180pt]{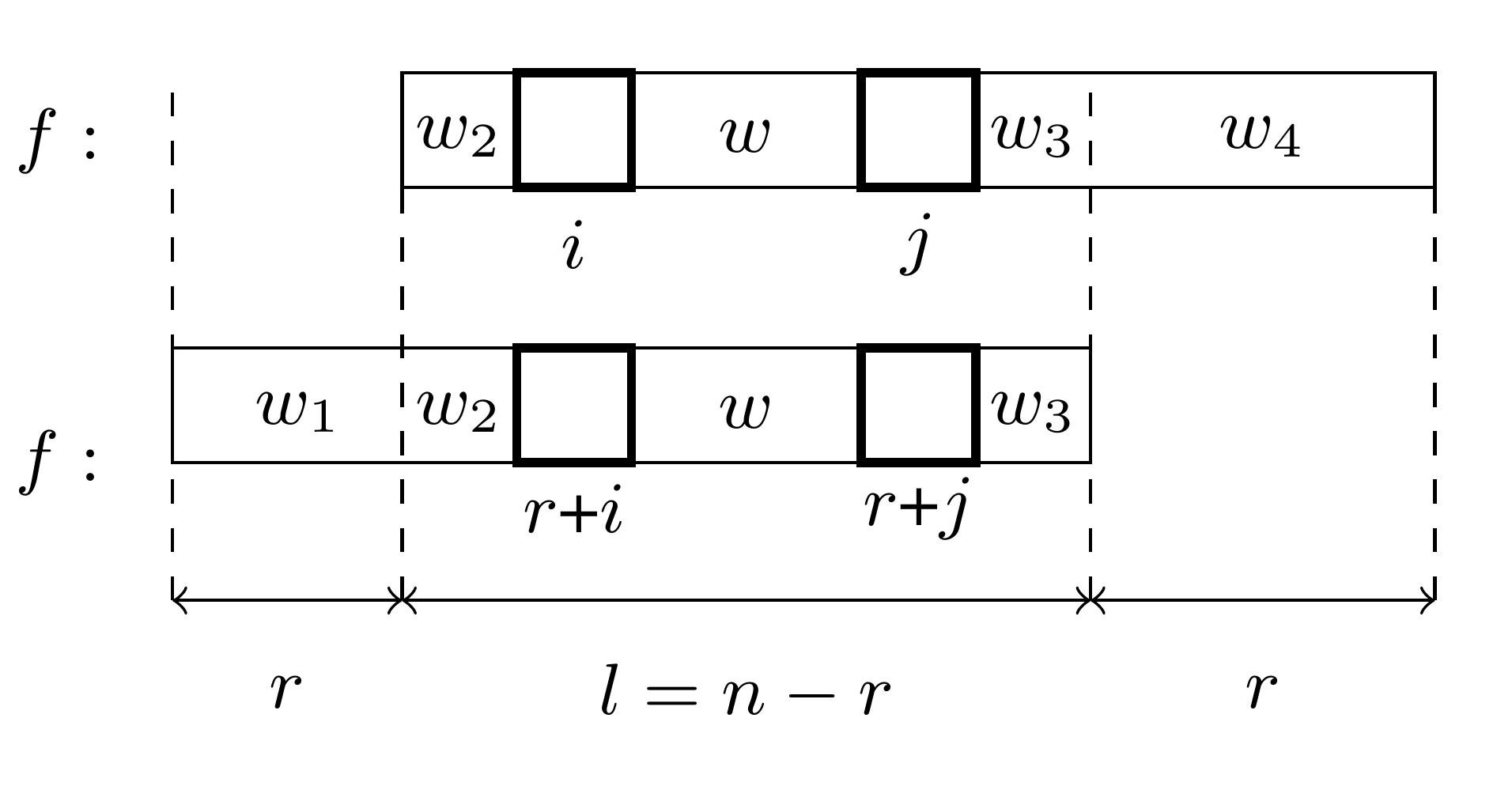}
\end{subfigure}
\hspace{1.4cm}
\begin{subfigure}{0.4\textwidth}
     \includegraphics[width=185pt]{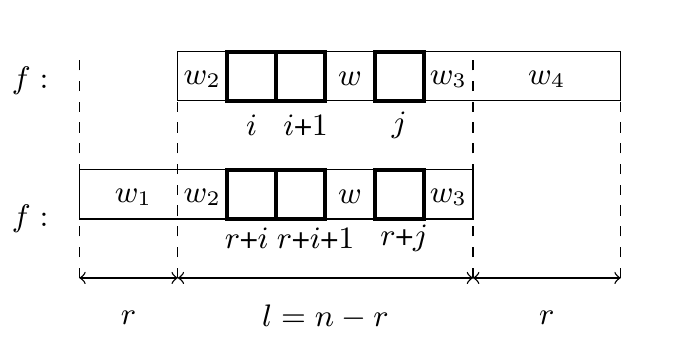} 
\end{subfigure}
    \caption{A word $f$ and its $2$-tilde-error overlap of type RR (a)  and SR (b)}\label{f-overlap}
\end{center}
\end{figure}

In other words, if $f$ has a $q$-tilde-error overlap of length $l$ then there exists a minimal tilde-transformation $\tau$ from ${\rm pre}_l(f)$ to ${\rm suf}_l(f)$ of length $q$. In the sequel, when $q=2$, in order to specify the kind of errors, a $2$-tilde-error overlap is referred to be of type RR if $\tau$
consists of two replacements, of type SS in case of two swaps, of type RS in case of replacement and swap, and 
of type SR in case of swap and replacement. If the two errors are in positions $i$ and $j$, with $i<j$ and we say that $f$ has a $2$-tilde-error overlap in $i$ and $j$ or, equivalently, that $i$ and $j$ are the {\em error positions} of the $2$-tilde-error overlap.  

Let $f\in \{0, 1 \}^n$ have a $2$-tilde-error overlap in positions $i$ and $j$ with $i<j$, of shift $r$ and length $l=n-r$. The following situations can occur (see Fig.\ref{f-overlap}), for some $w_1, w_2, w_3, w_4 \in \{0,1\}^*$ and $\vert w_1 \vert=r$.

\begin{description}
    \item RR: $f = w_2 \bm{f[i]} w \bm{f[j]} w_3w_4 = w_1 w_2 \bm{\overline{f[i]}} w \bm{\overline{f[j]}} w_3$
    \item SR: $f = w_2 \bm{f[i]f[i+1]} w \bm{f[j]} w_3w_4 = w_1 w_2 \bm{f[i+1]f[i]} w \bm{\overline{f[j]}} w_3$
    \item RS: $f = w_2 \bm{f[i]} w \bm{f[j]f[j+1]} w_3w_4 = w_1 w_2 \bm{\overline{f[i]}} w \bm{f[j+1]f[j]} w_3$
    \item SS: $f = w_2 \bm{f[i]f[i+1]} w \bm{f[j]f[j+1]} w_3w_4 = w_1 w_2 \bm{f[i+1]f[i]} w \bm{f[j+1]f[j]} w_3$
\end{description}

\noindent If $w=\epsilon$ we say that the two errors are {\em adjacent}. In particular, in the case of a $2$-tilde-error overlap in positions $i$ and $j$, of type RR and RS, the two errors are adjacent if $j=i+1$. Note that in case of adjacent errors of type RR with $f[i]\neq f[i+1]$, we have a $1$-tilde-error overlap that is a swap and that we call of type S. For 2-tilde error overlap of type SR and SS in positions $i$ and $j$, the two errors are adjacent if $j=i+2$. 

\begin{remark}\label{r-swap-a-cavallo}
Let $f\in\{0,1\}^n$ be a tilde-non-isometric word and $(u, v)$, with $u, v\in \Sigma^m$, be a pair of tilde-witnesses for $f$, with minimal $d={\rm dist}_\sim(u, v)$  among all pairs of tilde-witnesses of length $m$.

Let $\{O_{i_1}, O_{i_2}, \dots , O_{i_d}\}$ be the set of operations of a minimal tilde-transformation from $u$ to $v$,	
$1\leq i_1< i_2< \dots  < i_d\leq m$.
	Then, for any $j=1, 2, \dots, d-1$, $O_{i_j}(u)$ has an occurrence of $f$ in the interval $[k_j .. (k_j+n-1)]$, which contains at least one position modified by $O_{i_j}$.
In fact, if $O_{i_j}$ is a swap operation then it changes two positions at once, positions $i_j$ and $i_j+1$,  and the interval $[k_j .. (k_j+n-1)]$ may contain both positions or just one. Note that when only one position is contained in the interval, such position is at the boundary of the interval.  This means that,
although an error in a position at the boundary of a given interval may appear as caused by a replacement, this can be actually caused by a hidden swap involving positions over the boundary.
\end{remark}

\begin{proposition} \label{p_2_error_overlap}
    If $f \in\{0,1\}^n$ is tilde-non-isometric then 
     \begin{enumerate}
 \item either $f$ has a 1-tilde-error overlap of type S 
 \item or $f$ has a $2$-tilde-error overlap.
 \end{enumerate}
\end{proposition}
\begin{proof}
Let $f$ be a tilde-non-isometric word, $(u, v)$ be a pair of tilde-witnesses for $f$, and $\{O_{i_1}, O_{i_2}, \dots , O_{i_d}\}$ as in Remark \ref{r-swap-a-cavallo}.
Then, for any $j=1, 2, \dots, d-1$, $O_{i_j}(u)$ has an occurrence of $f$ in the interval $[k_j .. k_j+n-1]$, which contains at least one position modified by $O_{i_j}$. Note that,
 this occurrence of $f$ must disappear in a tilde-transformation from $u$ to $v$, because $v$ is $f$-free. Hence, the interval $[k_j .. k_j+n-1]$ contains a position modified by another operation in $\{O_{i_1}, O_{i_2}, \dots , O{i_d}\}$.
By the pigeonhole principle, there exist $s,t \in \{i_1, i_2, \dots i_d\}$, such that $O_s(u)$ has an occurrence of $f$ in $[k_s .. k_s+n-1]$ that contains at least one position modified by $O_{t}$ and $O_t(u)$ has an occurrence of $f$ in $[k_t .. k_t+n-1]$ that contains at least one position modified by $O_{s}$. Without loss of generality,  suppose that $k_s <k_t$. The intersection of $[k_s .. k_s+n-1]$ and $[k_t .. k_t+n-1]$ intercepts a prefix of $f$ in $O_t(u)$ and a suffix of $f$ in $O_s(u)$ of some length $l$. 
 Such an intersection can contain either two, or three, or four among the positions modified by $O_{s}$ and $O_{t}$, of which at least one is modified by $O_{s}$ and at least one by $O_{t}$.

 Consider the case that the intersection of $[k_s .. k_s+n-1]$ and $[k_t .. k_t+n-1]$  
contains two among the positions modified by $O_s$ and $O_t$, and denote them $i$ and $j$, with $1\leq i< j\leq l$.
If the positions are not adjacent, 
then $f$ has a $2$-tilde-error overlap (of type RR). 
Otherwise, if 
 $f[i]\neq f[i+1]$  then $f$ has a 1-tilde-error overlap of type S. If $f[i]=f[i+1]$  then $f$ has a  $2$-tilde-error overlap (of type RR).

Suppose that the intersection of $[k_s .. k_s+n-1]$ and $[k_t .. k_t+n-1]$  contains three among the positions modified by $O_{s}$ and $O_{t}$. In this case, at least one of the two operations must be a swap; suppose $O_s$ is a swap.
Then, $O_t$ could be either a replacement on the 
third position, or a swap if the third position is at the boundary of $[k_t .. k_t+n-1]$. In any case, $f$ has a $2$-tilde-error overlap (of type SR or SS).

Suppose now that the intersection of $[k_s .. k_s+n-1]$ and $[k_t .. k_t+n-1]$  contains four among the positions  modified by $O_{s}$ and $O_{t}$. In this case, each of $O_s$ and $O_t$ involves two positions, 
and $f$ has a $2$-tilde-error overlap of type SS. 

\end{proof}

\section{Construction of tilde-witnesses}\label{s-witness}
As already discussed in Section \ref{s-iso-e-2eo}, in order to prove that a word is tilde-non-isometric it is sufficient to exhibit a pair of tilde-witnesses.
Proposition \ref{p_2_error_overlap} states that if a word is tilde-non-isometric then it has either a $1$-tilde-error overlap of type swap or a $2$-tilde-error overlap. In this section we show the construction of tilde-witnesses for a word, starting from its error overlaps.
Let us start with the case of a $1$-tilde-error overlap.

\begin{proposition}
    If $f$ has a $1$-tilde-error overlap of type S, then it is tilde-non-isometric.
\end{proposition}
\begin{proof}
Let $f$ have a $1$-tilde-error overlap in position $i$ of type S with shift $r$. The pair $(u,v)$ with:
    $$u={\rm pre}_r(f)R_i(f) \hspace{1cm}v={\rm pre}_r(f)R_{i+1}(f)$$
is a pair of tilde-witnesses for $f$. In fact, one can prove that they satisfy the conditions in Definition\ref{d-witnesses}.

\end{proof}

\begin{example}
    The word $f=101$ has a $1$-tilde-error overlap of type S in position $1$ therefore it is tilde-non-isometric. In fact, following the proof of previous proposition,  the pair $(u,v)$ with $u=1001$ and $v=1111$ is a pair of tilde-witnesses.
\end{example}

Let us now introduce some special words which often will serve as tilde-witnesses.
Let $f\in \{0, 1 \}^n$ have a $2$-tilde-error overlap of shift $r$ in positions $i$ and $j$, then
 \begin{equation} \label{eq:alfa_beta_tilde}
     \tilde\alpha_r= {\rm pre}_r(f) O_i(f) \,\, \mbox{and} \,\, \tilde\beta_r= {\rm pre}_r(f) O_j(f)
  \end{equation}

As an example, using the previous notations for errors of type SR we have that
\begin{equation} \label{eq:alfa_beta_tilde_SR} 
\tilde\alpha_r(f)=w_1 w_2 f[i+1]  f[i] w f[j] w_3w_4 \,\, \mbox{and} \,\,
\tilde\beta_r(f)=w_1 w_2f[i] f[i+1] w \overline{f[j]} w_3w_4 
\end{equation}

\begin{lemma}\label{l-alfa-tilde-f-free}
    Let $f\!\in\! \{0, 1 \}^n$ have a $2$-tilde-error overlap of shift $r$, then $\tilde\alpha_r(f)$ is $f$-free.
\end{lemma}

\begin{proof}
Suppose that  $f\in \{0, 1 \}^n$ has a $2$-tilde-error overlap. If it is of type RR then $\tilde\alpha_r(f)$ is $f$-free by Claim 1 of Lemma 2.2 in \cite{Wei17}, also in the case of adjacent errors. If it is of type SR, of shift $r$ in positions $i$ and $j$, with $i<j$, then, by Equation (\ref{eq:alfa_beta_tilde}), we have $\tilde\alpha_r(f)=w_1S_i(f)$ then $\tilde\alpha_r[r+k]=f[k]$, for any $1 \leq k \leq n$, with $k \neq i$ and $k \neq i+1$. If $f$ occurs in $\tilde\alpha_r$ in position $r_1+1$ we have that $1 < r_1 < r$ (if $r_1=1$ then $f[i]=f[i+1]$ and there is no swap error at position $i$) and $\tilde\alpha_r[r_1+1 \ldots r_1+n]=f[1 \ldots n]$. Finally, by Equation (\ref{eq:alfa_beta_tilde_SR}), we have that $\tilde\alpha_r[k]=f[k]$, for $k \neq r+j$. In conclusion, we have that $f[i]=\tilde\alpha_r[r_1+i]=f[r_1+i]$ (trivially, $r_1+i \neq r+j$). Furthermore $f[r_1+i]=\tilde\alpha_r[r+r_1+i]$ ($r_1+i \neq i$ and $r_1+i \neq i+1$ because $r_1 >1$). But $\tilde\alpha_r[r+r_1+i]=f[r+i] $ then we have the contradiction that $f[i]=f[r+i]$. If the $2$-tilde-error is of type RS, SS the proof is similar. For clarity, note that, also in the case of adjacent errors, supposing that $f$ occurs in $\tilde\alpha_r$ leads to a contradiction in $f[i]$ that is not influenced by $j$.

\end{proof}

Note that while $\tilde\alpha_r$ is always $f$-free, $\tilde\beta_r$ is not. Indeed, the property  $\tilde\beta_r$  not $f$-free is related to a condition on the overlap of $f$. We  
 give the following definition. 

\begin{definition}
    Let $f\in \{0, 1 \}^n$ and consider a $2$-tilde-error overlap of $f$, with shift $r$ and  error positions $i, j$, with $1\leq i< j\leq n-r$. 
    The $2$-tilde-error overlap satisfies $Condition^{\sim}$ if it is of type RR or SS and:
    
    \begin{center}
	\hspace{1 cm}
	$\left\{
	\begin{array}{ll}
 
			r \ \ is\ even &\\
			j-i = r/2 & \\
				f[i .. (i+r/2-1)] =f[j .. (j+r/2 -1)] & \\
		\end{array}
		\right. $
    \hspace{1 cm}	(\textbf{\textit{ Condition$^{\sim}$} })
    \end{center}
\end{definition}

\begin{lemma}\label{l-beta-tilde-f-free}
	Let $f\in \{0, 1 \}^n$ have a $2$-tilde-error overlap of shift $r$, then $\tilde\beta_r(f)$ is not $f$-free iff the  $2$-tilde-error overlap
	satisfies $Condition^{\sim}$.
\end{lemma}

\begin{proof} 
Suppose that  $f\in \{0, 1 \}^n$ has a $2$-tilde-error overlap that satisfies $Condition^{\sim}$. 
{Note that a $2$-tilde-error overlap of type RS or SR cannot satisfy $Condition^{\sim}$.}
Now, if
 the $2$-tilde-error overlap is of type RR, then the fact that $\tilde\beta_r(f)$ is not $f$-free can be shown as in the proof of Claim 2 of Lemma 2.2 in \cite{Wei17}. If 
the $2$-tilde-error overlap is of type SS, then that proof must be suitably modified. More precisely, let $i, j$, with $1\leq i< j\leq n-r$, be the error positions of the $2$-tilde-error overlap of shift $r$ that satisfies $Condition^{\sim}$. 

Let $f[i]=f[j]=x$, $f[i+r]=f[j+r]=\overline{x}$, $f[i+1]=f[j+1]=\overline{x}$ and $f[i+1+r]=f[j+1+r]=x$. 

It is possible to show that, for some $k_1, k_2 \geq 0$, we can write 

\begin{center}
$f= \rho (uw)^{k_1}uwuw\overline{u}w\overline{u}(w\overline{u})^{k_2}\sigma$  
\end{center}

where $u=x\overline{x}$, $w=f[i+2 .. j-1]$ ($w$ is empty, if $j=i+2$) and
$\rho$ and $\sigma$ are, respectively, a suffix and a prefix of $w$.
Now , we have 
\begin{center}
$\tilde\beta_r(f)=\rho (uw)^{k_1+1}uwuw\overline{u}w\overline{u}(w\overline{u})^{k_2+1}\sigma$
\end{center}
 and, hence, $\tilde\beta_r(f)$ is not $f$-free.

 Assume now that $\tilde\beta_r(f)$ is not $f$-free and suppose that a copy of $f$ occurs in $\tilde\beta_r(f)$ at position $r_1+1$. A reasoning similar to the one used in the proof of Lemma  \ref{l-alfa-tilde-f-free}, shows that, if  $i$ and $j$ are the error positions, then  
 $j-i = r_1$ and $j-i = r-r_1$. Hence $r=2 r_1$ is even and $j-i=r/2$. Therefore, $f[i+t]=f^b[i+t]= f[i+t+r/2]= f[j+t]$, for  $0 \leq t \leq r/2$, i.e. $f[i .. (i+r/2-1)] =f[j .. (j+r/2 -1)$ and the $2$-tilde-error overlap satisfies $Condition^{\sim}$. 

\end{proof}

In the rest of the section we deal with the construction of tilde-witnesses in  the 
 case of $2$-tilde-error overlaps. 
 We distinguish the cases of non-adjacent  and  adjacent errors.  Non-adjacent errors can be dealt with standard techniques, while the case of adjacent ones may show new issues. For example,
when $f$ has a $2$-error-overlap of type SR with error block $\bm{101}$ (aligned with $010$) then it can
be also considered  of type RS.

 Moreover, note that all the adjacent pairs of errors can be  listed as follows,
 up to complement and reverse. A 2-error overlap of type SS may have (error) block $\bm{1010}$ or $\bm{1001}$; of type SR or RS may have 
  block $\bm{100}$, $\bm{101}$ or $\bm{110}$;
   of type RR block $\bm{11}$ (block $\bm{10}$ aligned with $01$ corresponds to one swap). Note  that, for some error types, we need also to distinguish sub-cases related to  the different characters adjacent to those error blocks. We collect all the cases in the following proposition. 
For lack of space, the proof is detailed only in the  case 2. In the remaining cases, the proofs are 
sketched  by exhibiting a pair of words that can be shown to be a pair of tilde-witnesses.
\begin{theorem}\label{t-main}
Let $f \in \{0,1\}^n$. Any of the  following conditions, up to complement and reverse, is sufficient for $f$ being tilde-non-isometric.
\begin{enumerate}
    \item\label{case1}$f$ has a $2$-tilde-error overlap with not adjacent error positions
    \item\label{case2} $f$ has a $2$-tilde-error overlap of type SS with adjacent error positions
    \item\label{case3} $f$ has a $2$-tilde-error overlap with block $\bm{101}$ (of type SR or RS)
    \item\label{case4} $f$ has a $2$-tilde-error overlap with block $\bm{100}$ (of type SR or RS) in the particular case that $f=x\bm{100}1z=yx\bm{011}$, for some $x,y,z \in \{0,1\}^*$ 
    \item\label{case5} $f$ has a $2$-tilde-error overlap RR in the particular case that $f$ starts with $\bm{110}$ and ends with $\bm{100}$.
\end{enumerate}

\end{theorem}

\begin{proof}
We provide, for each case in the list, a pair of tilde-witnesses for $f$.
\smallskip
\\ \noindent
{\em Case 1.}
If the  $2$-tilde-error overlap does not
	satisfy $Condition^{\sim}$, following Definition \ref{d-witnesses}, one can prove that the pair $(\tilde\alpha_r, \tilde\beta_r)$ as in Equation (\ref{eq:alfa_beta_tilde}) is a pair of tilde-witnesses for $f$. {Otherwise,  one can prove that $(\tilde\eta_r, \tilde\gamma_r)$ with $\tilde\eta_r={\rm pre}_r(f) O_i(f){\rm suf}_{r/2}(f)$ and $\tilde\gamma_r={\rm pre}_r(f) O_j(O_t(f)){\rm suf}_{r/2}(f)$ is a pair of tilde-witnesses for $f$. }
	
	\smallskip
  \noindent {\em Case 2.}
Proved in Lemma \ref{l-incucchiati_SS}.

\noindent {\em Case 3.}
We have $f=w_2\bm{101}w_3w_4=w_1w_2\bm{010}w_3$, for some $w_1, w_2, w_3, w_4 \in \{0,1\}^*$. The  pair $(\tilde\alpha_r,\tilde{\beta}_r)$, with
$\tilde\alpha_r=w_1w_2\bm{011}w_3w_4$ and  $\tilde\beta_r=w_1w_2\bm{100}w_3w_4$, is a pair of tilde-witnesses, following Definition \ref{d-witnesses}.

\noindent {\em Case 4.} We have $f=w_2\bm{100}1w_3=w_1w_2\bm{011}$, for some $w_1, w_2, w_3\in \{0,1\}^*$. In this case we need a different technique to construct the pair of tilde-witnesses $(\tilde\alpha_r,\B_r)$. We set 
$\tilde\alpha_r=w_1w_2\bm{0101}w_3$ and  $\B_r=w_1w_2\bm{1010}w_3$.  Here we prove that $\B_r$ is f-free. Indeed, suppose that a copy of $f$ occurs in  $\B_r$ starting from position $r_1$. Some considerations, related to the definition of $\B_r$ and to the structure of $f$, show that either $r_1=2$ or $r_1=3$, and one can prove that this leads to a contradiction.

 \noindent {\em Case 5.}
We have
$f=\bm{110}w_1=w_2\bm{100}$, for some $w_1,w_2 \in \{0,1\}^*$. 
By following Definition \ref{d-witnesses}, one can prove that the pair $(\tilde\alpha_r,\B_r)$, with $\tilde\alpha_r=w_2\bm{1010}w_1$ and $\B_r=w_2\bm{0101}w_1$ is a pair of tilde-witnesses. Remark that, in such a case, the pair $(\tilde\alpha_r,\tilde\beta_r)$ of Equation \ref{eq:alfa_beta_tilde} is not a pair of tilde-witnesses because ${\rm dist}_\sim(\tilde\alpha_r,\tilde\beta_r)=1$. 

\end{proof}

Let us prove  in details that Case 2. of previous theorem holds.
\begin{lemma}\label{l-incucchiati_SS}
    If $f$ has a $2$-tilde-error overlap of type SS, where the errors are adjacent, then $f$  is tilde-non-isometric.
\end{lemma}
\begin{proof}
Let $f\in \{0, 1 \}^n$ have a $2$-tilde-error overlap of shift $r$ and type SS, where the errors are adjacent. Then, two cases can occur (up to complement):

\noindent 
{\em Case 1:} $f=w_2\bm{1010}w_3w_4=w_1w_2\bm{0101}w_3$, $\vert w_1 \vert=r$
\\  
If the $2$-tilde-error overlap does not satisfy  Condition$^{\sim}$, then $(\tilde\alpha_r,\tilde{\beta}_r)$, with
$\tilde\alpha_r=w_1w_2\bm{0110}w_3w_4$ and  $\tilde\beta_r=w_1w_2\bm{1001}w_3w_4$,
is a pair of tilde-witnesses,
 following Definition \ref{d-witnesses}. In fact:
\begin{enumerate}
   \item $\tilde\alpha_r$ is $f$-free thanks to Lemma \ref{l-alfa-tilde-f-free} and $\tilde{\beta}_r$ is $f$-free, by Lemma \ref{l-beta-tilde-f-free}. 
    \item ${\rm dist}_\sim(\tilde\alpha_r, \tilde{\beta}_r)=2$, straightforward.
   \item a minimal tilde-transformation from $\tilde\alpha_r$ to $\tilde{\beta}_r$ consists of two swaps $S_i$ and $S_j$ with $i=\vert w_1w_2\vert+1$ and $j=i+2$. If $S_i$ is applied to $\tilde\alpha_r$ as first operation, then $f$ appears as a suffix, whereas if $S_j$ is applied first to $\tilde\alpha_r$, then $f$ appears as a prefix.
\end{enumerate}
If the $2$-tilde-error overlap satisfies Condition$^{\sim}$, then  $w_3=\bm{10}w^{\prime}_3$ and, following Definition \ref{d-witnesses}, $(\tilde{\eta}_r, \tilde{\gamma}_r)$ is a pair of tilde-witnesses, where $\tilde{\eta}_r=w_1w_2\bm{011010}w^{\prime}_3w_4w_5$ and $\tilde{\gamma}_r=w_1w_2\bm{100101}w^{\prime}_3w_4w_5$, with $w_5=suf_{r/2}(f)$. 
\begin{enumerate}

   \item One can prove that $\tilde{\eta}_r$ and $\tilde{\gamma}_r$ are $f$-free.
   \item ${\rm dist}_\sim(\tilde{\eta}, \tilde{\gamma})=3$.
 \item a minimal tilde-transformation from $\tilde{\eta}_r$ to $\tilde{\gamma}_r$ consists of three swap operations $S_i$, $S_j$ and $S_t$ with $i=\vert w_1w_2\vert+1$, $j=i+2$, $t=j+2$. 
 \\
 If $S_i$ is applied to $\tilde{\eta}_r$ as first operation, then $f$ occurs at position $|w_1|+1$, 
 if $S_j$ is applied first then 
 $f$ appears as a prefix,
 whereas if $S_t$ is applied first then $f$ appears as a suffix.   
\end{enumerate}

\noindent
{\em Case 2:}  $f=w_2\bm{1001}w_3w_4=w_1w_2\bm{0110}w_3,$ $\vert w_1 \vert=r$
\\
The pair $(\tilde{\alpha}_r,\tilde{\beta}_r)$, with
$\tilde{\alpha}_r=w_1w_2\bm{0101}w_3w_4$ and  $\tilde{\beta}_r=w_1w_2\bm{1010}w_3w_4$,
is a pair of tilde-witnesses,
 following Definition \ref{d-witnesses}. In such a case, by Lemma \ref{l-beta-tilde-f-free},  $\tilde{\beta}_r$ is $f$-free. In fact, the Condition$^\sim$ never holds, since $f[i]$ is different from $f[j]$. 

\end{proof}

 The following example uses Lemma \ref{l-incucchiati_SS}.

\begin{example}\label{e-2swap}
The word $f=\bm{1001}0110=1001\bm{0110}$ has a $2$-tilde-error overlap of type SS and shift $r=4$ in positions $1,3$. By Lemma \ref{l-incucchiati_SS}, the pair $(\tilde{\alpha}_4, \tilde{\beta}_4)$ with $\tilde{\alpha}_4=1001\bm{0101}0110$ and $\tilde{\beta}_4=1001\bm{1001}0110$ is a pair of tilde-witnesses. Then $f$ is tilde-non-isometric.
Note that $f$ is Ham-isometric.
\end{example}

Theorem \ref{t-main} lists the conditions for a word $f$ being tilde non-isometric and the proof provides all the corresponding pairs of witnesses. The  construction of $(\tilde\alpha_r, \tilde\beta_r)$ and $(\tilde\eta_r, \tilde\gamma_r)$, used so far, is inspired by an analogous construction for the Hamming distance (cf. \cite{Wei17}) and it is here adapted to the tilde-distance. On the contrary, in the cases \ref{case4} and \ref{case5} a new construction is needed because the usual pair of witnesses does not satisfy any more Definition \ref{d-witnesses}. 
{The construction of 
$\B_r$ is peculiar of the tilde-distance. It solves the situation expressed in Remark \ref{r-swap-a-cavallo} when
a mismatch error 
may appear as caused by a replacement, but it is actually caused by a hidden swap involving adjacent positions.
} 

 \smallskip
In conclusions, the swap and mismatch distance we adopted in this paper opens up new scenarios and presents interesting new situations that surely deserve further investigation.

\bibliographystyle{plain}
  \bibliography{references}

 \end{document}